\documentclass{article}

\usepackage{a4wide}
\usepackage{authblk}
\usepackage{url}
\usepackage{breakcites}
\usepackage{hyperref}

\usepackage{inconsolata}

\usepackage{url,hyperref}
\usepackage[utf8]{inputenc}
\usepackage[T1]{fontenc}
\usepackage{xspace}
\usepackage{color,soul}
\usepackage{amsmath,amssymb,amsthm}
\usepackage{cite}

\newtheorem{theorem}{Theorem}

\usepackage{tikz}
\usetikzlibrary{calc,positioning} % above of, etc.
\usetikzlibrary{arrows,shapes}
\usetikzlibrary{decorations.pathmorphing,decorations.markings}
\usetikzlibrary{patterns}

\usepackage[frozencache]{minted}
\def\coqin#1{\mintinline{ssr}{#1}}

\usepackage{textcomp}

\def\newterm#1{{\sl #1}}
\def\coq{\textsc{Coq}}
\def\monae{\textsc{Monae}}
\def\ssreflect{\textsc{SSReflect}}
\def\mathcomp{\textsc{MathComp}}
\def\infotheo{\textsc{InfoTheo}}
\def\paramcoq{\textsc{ParamCoq}}

\def\join{\coqin{Join}}
\def\ret{\coqin{Ret}}

\def\us{\char`\_}

\setminted{fontsize=\small}
\renewcommand{\paragraph}[2]{\noindent\textbf{#2}~}

\definecolor{setcolor}{HTML}{228B22}
\def\Set{{\color{setcolor}{\tt Set}}}
\def\impredicativeset{{\tt -impredicative-set}}

\title{Extending Equational Monadic Reasoning with Monad Transformers}

\author[1]{Reynald  Affeldt}
\affil[1]{National Institute of Advanced Industrial Science and Technology, Tokyo, Japan}

\author[2]{David Nowak}
\affil[2]{Univ. Lille, CNRS, Centrale Lille, UMR 9189 CRIStAL, F-59000 Lille, France}

\begin{document}

\maketitle

\begin{abstract}
  There is a recent interest for the verification of monadic programs
  using proof assistants.
  This line of research raises the question of the integration of
  monad transformers, a standard technique to combine monads.
  In this paper, we extend Monae, a Coq library for monadic equational
  reasoning, with monad transformers and we explain the benefits of
  this extension.
  Our starting point is the existing theory of modular monad
  transformers, which provides a uniform treatment of operations.
  Using this theory, we simplify the formalization of models in Monae
  and we propose an approach to support monadic equational reasoning
  in the presence of monad transformers.
  We also use Monae to revisit the lifting theorems of modular monad
  transformers by providing equational proofs and explaining how to
  patch a known bug using a non-standard use of Coq that combines
  impredicative polymorphism and parametricity.
\\
keywords: monads, monad transformers, Coq, impredicativity, parametricity
\end{abstract}

\section{Introduction}
\label{sec:intro}

There is a recent interest for the formal verification of monadic
programs stemming from \newterm{monadic equational reasoning\/}: an
approach to the verification of monadic programs that emphasizes
equational
reasoning~\cite{gibbons2011icfp,gibbons2012utp,mu2019tr2,mu2019tr3,mu2020flops}.
In this approach, an effect is represented by an operator belonging to
an interface together with equational laws. The interfaces all inherit
from the type class of monads and the interfaces are organized in a
hierarchy where they are extended and composed.
There are several efforts to bring monadic equational reasoning to
proof
assistants~\cite{affeldt2019mpc,pauwels2019mpc,affeldt2020trustful}.

In monadic equational reasoning, the user cannot rely on the
\newterm{model\/} of the interfaces because the implementation of the
corresponding monads is kept hidden. The construction of models is
nevertheless important to avoid mistakes when adding equational
laws~\cite{affeldt2020trustful}. This means that a formalization of
monadic equational reasoning needs to provide tools to formalize
models.

In this paper, we extend an existing formalization of monadic
equational reasoning (called \monae{}~\cite{affeldt2019mpc}) with
\newterm{monad transformers}.
Monad transformers is a well-known approach to combine monads that is
both modular and practical~\cite{liang1995popl}. It is also commonly
used to write Haskell programs.
The interest in extending monadic equational reasoning with monad
transformers is therefore twofold: (1)~it enriches the toolbox to
build formal models of monad interfaces, and (2)~it makes programs
written with monad transformers amenable to equational reasoning.
 
In fact, the interest for a formal theory of monad transformers goes
beyond its application to monadic equational reasoning.
Past research advances about monad transformers could have benefited
from formalization.
For example, a decade ago, Jaskelioff identified a lack of uniformity
in the definitions of the liftings of operations through monad
transformers~\cite{jaskelioff2009esop}.
He proposed \newterm{modular monad transformers} which come with a
uniform definition of lifting for operations that qualify as
\newterm{sigma-operations} or their sub-class of \newterm{algebraic
  operations}.
Unfortunately, the original proposal in terms of System F$\omega$
was soon ruled out as faulty~\cite[Sect.~6]{jaskelioff2010}
\cite[p.~7]{jaskelioff2009phd} and its fix gave rise to a more
involved presentation in terms of (non-trivial) category
theory~\cite{jaskelioff2010}.
More recently, this is the comparison between monad transformers and
algebraic effects that attracts attention, and it connects back to
reasoning using equational laws (e.g.,
\cite[Sect.~7]{schrijvers2019haskell}).
This is why in this paper, not only do we provide examples of monad
transformers and applications of monadic equational reasoning, but
also formalize a theory of monad transformers.

\paragraph*{Contributions}
In this paper, we propose a formalization in the \coq{} proof
assistant~\cite{coq} of monad transformers.
This formalization comes as an extension of \monae{}, an existing
library that provides a hierarchy of monad
interfaces~\cite{affeldt2019mpc}.
The benefits of this extension are as follows.
\begin{itemize}
\item The addition of sigma-operations and of monad transformers to
  \monae{} improves the implementation of models of monads.
  These models are often well-known and it is tempting to define them
  in an ad hoc way. Sigma-operations help us discipline proof scripts
  and naming, which are important aspects of proof-engineering.
\item We illustrate with an example how to extend \monae{} to verify a
  program written with a monad transformer.
  Verification is performed by equational reasoning using equational
  laws from a monad interface whose model is built using a monad
  transformer.
\item We use our formalization of monad transformers to formalize the
  theory of lifting of modular monad transformers.
  Thanks to \monae{}, the main theorems of modular monad transformers
  can be given short formal proofs in terms of equational reasoning.
\end{itemize}

Regarding the theory of lifting of modular monad transformers, our
theory fixes the original presentation~\cite{jaskelioff2009esop}.
This fix consists in a non-standard use of \coq{} combining
impredicativity and parametricity (as implemented by \paramcoq{}~\cite{keller2012csl})
that allows for an encoding using
the language of the proof assistant and thus avoids the hassle of
going through a technical formalization of category theory
(which is how Jaskelioff fixed his original proposal).
It must be said that this was not possible at the time of the original
paper on modular monad transformers because parametric models of
dependent type theory were not known~\cite{atkey2014popl} (but were
``expected''~\cite{jaskelioffPC}).
We are therefore in the situation where formalization using a proof
assistant allowed for a fruitful revisit of pencil-and-paper proofs.

Regarding the benefit of extending \monae{} with sigma-operations and
monad transformers, we would like to stress that this is also a step
towards more modularity in our formalization of monadic equational
reasoning.
Indeed, one important issue that we have been facing is the quality of
our proof scripts.
Proof scripts that reproduce monadic equational reasoning must be as
concise as they are on paper.
Proof scripts that build models (and prove lemmas) should be
maintainable (to be improved or fixed easily in case of changes in the
hypotheses) and understandable (this means having a good balance
between the length of the proof script and its readability).
This manifests as mundane but important tasks such as factorization of
proof scripts, generalization of lemmas, abstraction of data
structures, etc.
From the viewpoint of proof-engineering, striving for modularity is
always a good investment because it helps in breaking the
formalization task into well-identified, loosely-coupled pieces.

\paragraph*{Outline}
In Sect.~\ref{sec:background}, we recall the main constructs of
\monae{}.
In Sect.~\ref{sec:extension}, we formalize the basics of modular monad
transformers: sigma-operations, monad transformers, and their
variants (algebraic operations and functorial monad transformers).
Section~\ref{sec:application} is our first application: we show with
an example how to extend \monae{} to verify a program written using
monad transformers.
In Sect.~\ref{sec:theorem19}, we use our formalization of monad
transformers to prove a first theorem about modular monad transformers
(namely, the lifting of algebraic operations) using equational
reasoning.
In Sect.~\ref{sec:lifting_operations}, we formalize (and fix) the main
theorem of modular monad transformers (namely, the lifting of
sigma-operations that are not necessarily algebraic along 
functorial monad transformers).
We review related work in Sect.~\ref{sec:related_work} and conclude in
Sect.~\ref{sec:conclusion}.

\section{Overview of the \monae{} Library}
\label{sec:background}

\monae{}~\cite{affeldt2019mpc} is a formal library implemented in the
\coq{} proof assistant~\cite{coq} to support monadic equational
reasoning~\cite{gibbons2011icfp}.
It takes advantage of the rewriting capabilities of the tactic
language called \ssreflect{}~\cite{ssrman} to achieve formal proofs by
rewriting that are very close to their pencil-and-paper counterparts.
\monae{} provides a hierarchy of monad interfaces formalized using the
methodology of \newterm{packed classes}~\cite{garillot2009tphols}.
Effects are declared as operations in interfaces together with
equational laws, and some effects extend others by (simple or
multiple) inheritance.
This modularity is important to achieve natural support for monadic
equational reasoning.

\def\Ret#1#2{{{\ret}^{\coqin{#1}}_{\coqin{#2}}}}
\def\Join#1#2{{{\join}^{\coqin{#1}}_{\coqin{#2}}}}

Let us briefly explain some types and notations provided by \monae{}
that we will use in the rest of this paper.
\monae{} provides basic category-theoretic definitions such as
functors, natural transformations, and monads. By default, they are
specialized to \coqin{UU0}, the lowest universe in the hierarchy of
\coq{} types, understood as a category\footnote{\monae{} also provides
  a more generic setting~\cite[file \coqin{category.v}]{monae} but we
  do not use it in this paper.}.
The type of functors is \coqin{functor}.
The application of a functor \coqin{F} to a function \coqin{f}
is denoted by \mintinline{ssr}{F # f}.
The composition of functors is denoted by the infix notation~\coqin{\O}.
The identity functor is denoted by \coqin{FId}.
Natural transformations from the functor \coqin{F} to the functor
\coqin{G} are denoted by \coqin{F ~> G}. Natural transformations are
formalized by their components (represented by the type \coqin{forall
  A, F A -> G A}, denoted by \coqin{F ~~> G}) together with the proof that
they are natural, i.e., the proof that they satisfy the following
predicate:
\begin{minted}{ssr}
Definition naturality (M N : functor) (m : M ~~> N) :=
  forall (A B : UU0) (h : A -> B), (N # h) \o m A = m B \o (M # h).
\end{minted}
(The infix notation \coqin{\o} is for function composition.)
Vertical composition of natural transformations is denoted by the
infix notation \coqin{\v}.  The application of a functor \coqin{F} to
a natural transformation \coqin{n} is denoted by \mintinline{ssr}{F ## n}.

The type of monads is \coqin{monad}, which inherits from the type
\coqin{functor}. Let \coqin{M} be of type \coqin{monad}. Then
\ret{} is a natural transformation \coqin{FId ~> M} and
\join{} is a natural transformation \coqin{M \O M ~> M}.
Using \ret{} and \join{}, we define the standard bind
operator with the notation~\coqin{>>=}.

In this paper, we show \coq{} proof scripts verbatim when it is
reasonable to do so.  When we write mathematical formulas, we keep the
same typewriter font, but, for clarity and to ease reading, we make
explicit some information that would otherwise by implicitly inferred
by \coq{}.
For example, one simply writes \ret{} or \join{} in proof scripts
written using \monae{} because it has been implemented in such a way
that \coq{} infers from the context which monad they refer to and
which type they apply to.  In mathematical formulas, we sometimes make
the monad explicit by writing it as a superscript of \ret{} or
\join{}, and we sometimes write the argument of a function application
as a subscript.  This leads to terms such as $\Ret{M}{A}$: the unit of
the monad \coqin{M} applied to some type~\coqin{A}.

See the online development for technical details (in particular,
\cite[file \coqin{hierarchy.v}]{monae}).

\section{Sigma-operations and Monad Transformers in \monae{}}
\label{sec:extension}

The first step is to formalize sigma-operations
(Sect.~\ref{sec:sigma-operations}) and monad transformers
(Sect.~\ref{sec:transf}).
We illustrate sigma-operations with the example of the model of the
state monad (Sect.~\ref{sec:statemodel}) and its get operation
(Sect.~\ref{sec:sigmaget}).

\subsection{Extending \monae{} with Sigma-operations}
\label{sec:sigma-operations}

Given a functor \coqin{E}, an \newterm{\coqin{E}-operation} for a monad
\coqin{M} (sigma-operation for short) is a natural transformation from
\coqin{E \O M} to \coqin{M}.
The fact that sigma-operations are defined in terms of natural
transformations is helpful to build models because it involves
structured objects (functors and natural transformations) already
instrumented with lemmas.
In other words, we consider sigma-operations as a disciplined way to
formalize effects.
For illustration, we explain how the get operation of the state monad
is formalized.

\subsubsection{Example: Model of the State Monad}
\label{sec:statemodel}

First we define a model \coqin{State.t} for the state monad (without
get and put for the time being).
We assume a type~\coqin{S} (line~\ref{line:typeS}) and define the
action on objects \coqin{acto} (line~\ref{line:actoS}), abbreviated
as~\coqin{M} (line~\ref{line:notationM}). We define the action on
morphisms \coqin{map} (line~\ref{line:mapS}) and prove the functor
laws (omitted here, see \cite[file \coqin{monad_model.v}]{monae} for
details).  This provides us with a functor \coqin{functor}
(line~\ref{line:functorS}, \coqin{Functor.Pack} and
\coqin{Functor.Mixin} are constructors from \monae{} and are named
after the packed classes methodology~\cite{garillot2009tphols}).
We define the unit of the monad by first providing its components
\coqin{ret_component} (line~\ref{line:ret0S}), and prove naturality
(line~\ref{line:retS_nat}, proof script omitted).  We then package this
proof to form a genuine natural transformation at line~\ref{line:retS}
(\coqin{Natural.Pack} and \coqin{Natural.Mixin} are constructors from
\monae{}).
We furthermore define \coqin{bind} (line~\ref{line:bindS}), prove the
properties of the unit and bind (omitted).  Finally, we call the
function \coqin{Monad_of_ret_bind} from \monae{} to build the monad
(line~\ref{line:monadS}):
\begin{minted}[fontsize=\small,numbers=left,xleftmargin=2em,escapeinside=77]{ssr}
(* in Module State *)
Variable S : UU0. 7\label{line:typeS}7
Definition acto := fun A => S -> A * S. 7\label{line:actoS}7
Local Notation M := acto. 7\label{line:notationM}7
Definition map A B (f : A -> B) (m : M A) : M B := 7\label{line:mapS}7
 fun (s : S) => let (x1, x2) := m s in (f x1, x2).
(* functor laws map_id and map_comp omitted *)
Definition functor := Functor.Pack (Functor.Mixin map_id map_comp). 7\label{line:functorS}7
Definition ret_component : FId ~~> M := fun A a => fun s => (a, s). 7\label{line:ret0S}7
Lemma naturality_ret : naturality FId functor ret_component. 7\label{line:retS_nat}7
(* proof script of naturality omitted *)
Definition ret : FId ~> functor := 7\label{line:retS}7
  Natural.Pack (Natural.Mixin naturality_ret).
Definition bind := fun A B (m : M A) (f : A -> M B) => uncurry f \o m. 7\label{line:bindS}7
(* proofs of neutrality of ret and of associativity of bind omitted *)
Definition t := Monad_of_ret_bind left_neutral right_neutral associative. 7\label{line:monadS}7
\end{minted}

\subsubsection{Example: The Get Operation as a Sigma-operation}
\label{sec:sigmaget}

By definition, for each sigma-operation we need a functor. The functor
corresponding to the get operation is defined below as
\coqin{Get.func} (line~\ref{line:getfunc}): \coqin{acto} is the action
on the objects, \coqin{actm} is the action on the morphisms (the
prefix~\coqin{@} disables implicit arguments in \coq{}):
\begin{minted}[fontsize=\small,numbers=left,xleftmargin=2em,escapeinside=77]{ssr}
(* in Module Get *)
Variable S : UU0.
Definition acto X := S -> X.
Definition actm (X Y : UU0) (f : X -> Y) (t : acto X) : acto Y := f \o t.
Program Definition func := Functor.Pack (@Functor.Mixin _ actm _ _). 7\label{line:getfunc}7
(*  proofs of the functors law omitted *)
\end{minted}
We then define the sigma-operation itself
(\coqin{StateOps.get_op} at line~\ref{line:getop}), which is a natural transformation from
\coqin{Get.func S \O M} to \coqin{M}, where \coqin{M} is the state
monad \coqin{State.t} built in Sect.~\ref{sec:statemodel}.  Note that this get operation
($\lambda s.\, k\,s\,s$, line~\ref{line:get}) is {\em not\/} the usual operation~\cite[Example~13]{jaskelioff2009esop}.
\begin{minted}[fontsize=\small,numbers=left,xleftmargin=2em,escapeinside=77]{ssr}
(* in Module StateOps *)
Variable S : UU0.
Local Notation M := (State.t S).
Definition get A (k : S -> M A) : M A := fun s => k s s. 7\label{line:get}7
Lemma naturality_get : naturality (Get.func S \O M) M get.
(* proof script of naturality omitted *)
Definition get_op : (Get.func S).-operation M := 7\label{line:getop}7
  Natural.Pack (Natural.Mixin naturality_get).
\end{minted}

\subsubsection{Example: Model of the Interface of the State Monad}

\monae{} originally comes with an interface \coqin{stateMonad} for the
state monad ({\em with} the get and put operations). It implements the
interface as presented by Gibbons and
Hinze~\cite[Sect.~6]{gibbons2011icfp}; it therefore expects the
operations to be the usual ones.  We show how to instantiate it using
the definition of sigma-operations.
First, we need to define the usual get from \coqin{StateOps.get_op}
(line~\ref{line:usual_get} below):
\begin{minted}[fontsize=\small,numbers=left,xleftmargin=2em,escapeinside=77]{ssr}
(* in Module ModelState *)
Variable S : UU0.
Local Notation M := (ModelMonad.State.t S).
Definition get : M S := StateOps.get_op _ Ret. 7\label{line:usual_get}7
\end{minted}
We do the same for the put operation (omitted).
We then build the model of interface of the state monad (with its
operations) using the appropriate constructors from \monae{}:
\begin{minted}[fontsize=\small]{ssr}
Program Definition state : stateMonad S := MonadState.Pack (MonadState.Class
  (@MonadState.Mixin _ _ get put _ _ _ _)).
(* proofs of the laws of get and put automatically discharged *)
\end{minted}

Similarly, using sigma-operations, we have formalized the operations
of the list, the output, the state, the environment, and the
continuation monads, which are the monads discussed along with modular
monad transformers~\cite[Fig.~1]{jaskelioff2009esop} (see~\cite[file
\coqin{monad_model.v}]{monae} for their formalization).

\subsection{The Sub-class of Algebraic Operations}
\label{sec:algebraic_operations}

An \coqin{E}-operation \coqin{op} for \coqin{M} is
\newterm{algebraic}~\cite[Def.~15]{jaskelioff2009esop} when
it satisfies the predicate \coqin{algebraicity} defined as follows in \coq{}
(observe the position of the continuation~\coqin{>>= f}):
\begin{minted}{ssr}
forall A B (f : A -> M B) (t : E (M A)),
  op A t >>= f = op B ((E # (fun m => m >>= f)) t).
\end{minted}
Algebraic operations are worth distinguishing because they lend
themselves more easily to lifting, and this result can be used to
define lifting for the whole class of sigma-operations (this is the
purpose of Sections~\ref{sec:theorem19}
and~\ref{sec:lifting_operations}).
We can check using \coq{} that, as expected, all the operations
discussed along with modular monad
transformers~\cite[Fig.~1]{jaskelioff2009esop} are algebraic except
for flush, local, and handle\footnote{In fact, we had to fix the
  output operation of the output monad.  Indeed, it is defined as
  follows in \cite[Example 32]{jaskelioff2009esop}:
$
{\sf output}((w,m) : W \times {\sf O}X) : {\sf O}X \,\hat{=}\, {\sf let\,}(x,w') = m {\sf \,in\,}(x, {\sf append}(w',w)).
$
We changed ${\sf append}(w',w)$ to ${\sf append}(w,w')$ to be able to prove algebraicity.}.

\subsubsection*{Example: the Get operation is Algebraic}

For example, the get operation of the state monad is algebraic:
\begin{minted}{ssr}
Lemma algebraic_get S : algebraicity (@StateOps.get_op S). Proof. by []. Qed.
\end{minted}
In the \coq{} formalization, we furthermore provide the type
\coqin{E.-aoperation M} (note the prefix ``\coqin{a}'') of an
\coqin{E.-operation M} that is actually algebraic. For example, here
is how we define the algebraic version of the get operation:
\begin{minted}[fontsize=\small]{ssr}
Definition get_aop S : (StateOps.Get.func S).-aoperation (ModelMonad.State.t S) :=
  AOperation.Pack (AOperation.Class (AOperation.Mixin (@algebraic_get S))).
\end{minted}

\subsection{Extending \monae{} with Monad Transformers}
\label{sec:transf}

Given two monads \coqin{M} and \coqin{N}, a \newterm{monad morphism}
\coqin{e} is a function of type \coqin{M ~~> N} such
that for all types \coqin{A}, \coqin{B} the following laws hold:
\begin{itemize}
\item
\begin{minted}{ssr}
e A \o Ret = Ret.                       (* MonadMLaws.ret *)
\end{minted}
\item
\begin{minted}{ssr}
forall (m : M A) (f : A -> M B),        (* MonadMLaws.bind *)
  e B (m >>= f) = e A m >>= (e B \o f).
\end{minted}
\end{itemize}

In \coq{}, we define the type of monad morphisms \coqin{monadM} that
implement the two laws above.
Monad morphisms are also natural transformations (this can be proved
easily using the laws of monad morphisms).
We therefore equip monad morphisms \coqin{e} with a canonical
structure of natural transformation. Since it is made canonical, \coq{}
is able to infer it in proof scripts but we need to make it
explicit in statements; we provide the notation \coqin{monadM_nt e}
for that purpose.

A \newterm{monad transformer} \coqin{t} is a function of type
\coqin{monad -> monad} with an operator \coqin{Lift} such that for
any monad~\coqin{M}, \coqin{Lift t M} is a monad morphism from
\coqin{M} to \coqin{t M}. Let \coqin{monadT} be the type of monad
transformers in \monae{}.
We reproduced all the examples of modular monad transformers (state,
exception, environment, output, continuation monad transformers,
resp.\ \coqin{stateT}, \coqin{exceptT}, \coqin{envT}, \coqin{outputT},
and \coqin{contT} in \cite[file \coqin{monad_transformer.v}]{monae}).

\subsubsection*{Example: The Exception Monad Transformer}

Let us assume given some type \coqin{Z : UU0} for exceptions and some
monad~\coqin{M}.  First, we define the action on objects of the monad
transformed by the exception monad transformer (the type \coqin{Z + X}
represents the sum type of the types~\coqin{Z} and~\coqin{X}):
\begin{minted}[fontsize=\small]{ssr}
Definition MX := fun X : UU0 => M (Z + X).
\end{minted}
We also define the unit and the bind operator of the transformed monad
(the constructors \coqin{inl}/\coqin{inr} inject a type into the
left/right of a sum type):
\begin{minted}[fontsize=\small]{ssr}
Definition retX X x : MX X := Ret (inr x).
Definition bindX X Y (t : MX X) (f : X -> MX Y) : MX Y :=
  t >>= fun c => match c with inl z => Ret (inl z) | inr x => f x end.
\end{minted}
Second, we define the monad morphism that will be returned by the lift
operator of the monad transformer. In \coq{}, we can formalize the
corresponding function by constructing the desired monad
assuming~\coqin{M}.  This is similar to the construction of the state
monad we saw in Sect.~\ref{sec:sigma-operations}.
We start by defining the underlying functor \coqin{MX_map}, prove the
two functor laws (let us call \coqin{MX_map_i} and \coqin{MX_map_o}
these proofs), and package them as a functor:
\begin{minted}[fontsize=\small]{ssr}
Definition MX_functor := Functor.Pack (Functor.Mixin MX_map_i MX_map_o).
\end{minted}
We then provide the natural transformation \coqin{retX_natural}
corresponding to \coqin{retX} and call the \monae{} constructor
\coqin{Monad_of_ret_bind} (like we did in
Sect.~\ref{sec:sigma-operations}):
\begin{minted}[fontsize=\small]{ssr}
Program Definition exceptTmonad : monad :=
  @Monad_of_ret_bind MX_functor retX_natural bindX _ _ _.
(* proofs of monad laws omitted *)
\end{minted}
Then we define the lift operation as a function that given a
computation \coqin{m} in the monad \coqin{M X} returns a computation
in the monad \coqin{exceptTmonad X}:
\begin{minted}[fontsize=\small]{ssr}
Definition liftX X (m : M X) : exceptTmonad X := m >>= (@RET exceptTmonad _).
\end{minted}
(The function \coqin{RET} is a variant of \coqin{Ret} better suited
for type inference here.)  We can finally package the definition of
\coqin{liftX} to form a monad morphism:
\begin{minted}{ssr}
Program Definition exceptTmonadM : monadM M exceptTmonad :=
  monadM.Pack (@monadM.Mixin _ _ liftX _ _).
(* proof of monad morphism laws omitted *)
\end{minted}

The exception monad transformer merely packages the monad morphism we
have just defined to give it the type~\coqin{monadT}:
\begin{minted}{ssr}
Definition exceptT Z := MonadT.Pack (MonadT.Mixin (exceptTmonadM Z)).
\end{minted}

One might wonder what is the relation between the monads that can be
built with these monad transformers and the monads already present in
\monae{}. For example, in Sect.~\ref{sec:sigma-operations}, we already
mentioned the \coqin{stateMonad} interface and we built a model for it
(namely, \coqin{ModelState.state}). On the other
hand, we can now, say, build a model for the identity monad (let us
call it \coqin{identity}) and build a model for that state monad as
\coqin{stateT S identity} (we have not provided the details of
\coqin{stateT}, see \cite{affeldt2019mpc}).  We can actually prove in
\coq{} that \coqin{stateT S identity} and \coqin{State.t} are {\em
  equal}\footnote{\cite[\coqin{Section
    instantiations_with_the_identity_monad}, file
  \coqin{monad_model.v}]{monae}}, so that no confusion has been
introduced by extending \monae{} with monad transformers.

\subsection{Functorial Monad Transformers}
\label{sec:fmt}

A \newterm{functorial monad transformer}~\cite[Def.~20]{jaskelioff2009esop} is a monad transformer \coqin{t} with
a function \coqin{h} (hereafter denoted by \coqin{Hmap t}) of type
\begin{minted}{ssr}
forall (M N : monad), (M ~> N) -> (t M ~> t N)
\end{minted}
such that
(1)~\coqin{h} preserves monad morphisms (the laws \coqin{MonadMLaws.ret} and \coqin{MonadMLaws.bind} seen in Sect.~\ref{sec:transf}),
(2)~\coqin{h} preserves identities and composition of natural transformations, and
(3)~\coqin{Lift t} is natural, i.e.,
\begin{minted}{ssr}
forall (M N : monad) (n : M ~> N) X, h M N n X \o Lift t M X = Lift t N X \o n X.
\end{minted}
Note that we cannot define the naturality of \coqin{Lift t} using the
predicate \coqin{naturality} we saw in Sect.~\ref{sec:background}
because it is restricted to endofunctors on~\coqin{UU0}.
Also note that Jaskelioff distinguishes monad transformers from
functorial monad transformers while Maillard defines monad transformers
as functorial by default~\cite[Def.~4.1.1]{maillard2019phd}.

\section{Application 1: Monadic Equational Reasoning in the Presence of Monad Transformers}
\label{sec:application}

We apply our formalization of monad transformers to the verification
of a recursive program combining the effects of state and
exception. We argue that this program is similar in style to what an
Haskell programmer would typically write with monad
transformers. Despite this programming style and the effects, the
correctness proof is by equational reasoning.

\subsection{Extending the Hierarchy}

\label{sec:exthier}

\begin{figure}
\includegraphics[width=14cm]{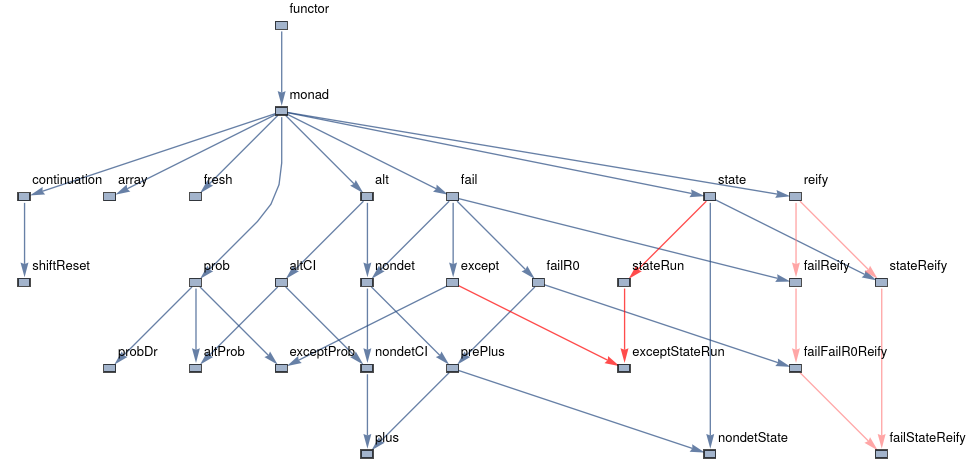}
\caption{Hierarchy of Monad Interfaces Provided by \monae{}}
\label{fig:hier}
\end{figure}

The first thing to do is to extend the hierarchy of
interfaces with \coqin{stateRunMonad} and
\coqin{exceptStateRunMonad} (Fig.~\ref{fig:hier}).

The interface \coqin{stateRunMonad} is a parameterized interface that
extends \coqin{stateMonad} with the primitive \coqin{RunStateT} and
its equations. Concretely, let \coqin{N} be a monad and \coqin{S} be
the type of states. When \coqin{m} is a computation in the monad
\coqin{stateRunMonad S N}, \coqin{RunStateT m s} runs \coqin{m} in a
state \coqin{s} and returns a computation in the monad~\coqin{N}.
There is one equation for each combination of \coqin{RunStateT} with
operations below in the hierarchy:
\begin{minted}{ssr}
RunStateT (Ret a)   s = Ret (a, s)
RunStateT (m >>= f) s = RunStateT m s >>= fun x => RunStateT (f x.1) x.2
RunStateT Get       s = Ret (s, s)
RunStateT (Put s')  s = Ret (tt, s')
\end{minted}
This is the methodology of packed classes that allows for the
overloading of the notations~\coqin{Ret} and \coqin{>>=} here. The
notation \coqin{.1} (resp.\ \coqin{.2}) is for the first (resp.\
second) projection of a pair.
The unique value of type~\coqin{unit} is~\coqin{tt}.
The operations \coqin{Get} and \coqin{Put} are the standard operations
of the state monad. Intuitively, given a monad \coqin{M} that inherits
from the state monad, \coqin{Get} is a computation of type \coqin{M S}
that returns the state and \coqin{Put} has type \coqin{S -> M unit}
and updates the state (see Sect.~\ref{sec:extension} for a model of
these operations).

The interface \coqin{exceptStateRunMonad} is the combination of
the operations and equations of \coqin{stateRunMonad} and
\coqin{exceptMonad}~\cite[Sect.~5]{gibbons2011icfp}\cite[file \coqin{hierarchy.v}]{monae} plus two
additional equations on the combination of \coqin{RunStateT} with
the operations of \coqin{exceptMonad}.
Recall that the operations of the exception monad are the computations
\coqin{Fail} of type \coqin{M A} and \coqin{Catch} of type \coqin{M A
  -> M A -> M A} for some type~\coqin{A} (which happens to be the type of
the state in this example); intuitively, \coqin{Fail} raises an exception
while \coqin{Catch} handles it.
\begin{minted}{ssr}
RunStateT Fail          s = Fail
RunStateT (Catch m1 m2) s = Catch (RunStateT m1 s) (RunStateT m2 s)
\end{minted}

Using our formalization of monad transformers presented in this paper,
it is then easy to build a model that validates those equations,
whereas in previous work we had to build a model from scratch each
time we were introducing a new combination of effects.

\subsection{Example: The Fast Product}

Now let us write a program and reason on it equationally.
First, we write a recursive function that traverses a list of natural numbers to compute their product, but fails in case a 0 is met. Intermediate results are stored in the state:
\begin{minted}{ssr}
Variables (N : exceptMonad) (M : exceptStateRunMonad nat N).
Fixpoint fastProductRec l : M unit :=
  match l with
  | [::] => Ret tt
  | 0 :: _ => Fail
  | n.+1 :: l' => Get >>= fun m => Put (m * n.+1) >> fastProductRec l'
  end.
\end{minted}
Then, the main function will catch an eventual failure.  If there is a
failure, then the result is~0, else the result is the value stored in
the state:
\begin{minted}{ssr}
Variables (N : exceptMonad) (M : exceptStateRunMonad nat N).
Definition fastProduct l : M _ :=
  Catch (Put 1 >> fastProductRec l >> Get) (Ret 0 : M _).
\end{minted}

To implement this algorithm in Haskell, we would use the state monad
transformer applied to the exception monad. It would then be necessary to
prefix each primitive of the exception monad with a
lifting operation (\texttt{lift} or \texttt{mapStateT2} in
Haskell). Here, we avoid this by using the hierarchy of interfaces,
and the use of monad transformers is restricted to the construction of
models for the interfaces.

The correctness states that the result of the fast product is always the same as a purely functional version:
\begin{minted}{ssr}
Lemma fastProductCorrect l n : evalStateT (fastProduct l) n = Ret (product l).
\end{minted}
where \coqin{evalStateT m s} is defined as \coqin{RunStateT m s >>=
  fun x => Ret x.1}.
This proposition is proved easily with a 10 lines proof script that
consists of an induction on \coqin{l}, rewriting with the equations
in \coqin{exceptStateRunMonad} and application of standard arithmetic
(see Appendix~\ref{sec:fpscript}).

Note that in this section we are dealing with the state monad
transformer applied to the exception monad, and that the last equation
in Sect.~\ref{sec:exthier} specifies that the state is
``backtracked'', i.e., if the state is modified in \coqin{m1} before
an exception occurs, then this change is forgotten before \coqin{m2}
is executed. This is usual in Haskell.
The alternative semantics without backtracking would be closer to,
say, OCaml, where the state is not be backtracked in case of an
exception.  Our program would behave the same way because it happens
that the exception handler ignores the state. However, we would need
to devise new equations to deal with \coqin{Fail} and \coqin{Catch}.

\section{Application 2: Formalization of the Lifting of an Algebraic Operation}
\label{sec:theorem19}

This section is an application of \monae{} extended with the
formalization of sigma-operations and of monad transformers of
Sect.~\ref{sec:extension}.
We prove using equational reasoning a theorem about the lifting of
algebraic operations along monad morphisms.
This corresponds to the theorem that concludes the first of part of
the original paper on modular monad
transformers~\cite[Sect.2--4]{jaskelioff2009esop}.

\newsavebox\mybox
\begin{lrbox}{\mybox} \mintinline{ssr}{#} \end{lrbox}
\def\mysharp{{\usebox{\mybox}}}

In the following \coqin{M} and \coqin{N} are two monads.
Given an \coqin{E}-operation \coqin{op} for \coqin{M} and a monad
morphism \coqin{e} from \coqin{M} to \coqin{N}, a \newterm{lifting} of
\coqin{op} (to \coqin{N}) along \coqin{e} is an \coqin{E}-operation
\coqin{op'} for \coqin{N} such that for all \coqin{X}:
$$
\coqin{e}_{\coqin{X}} \circ \coqin{op}_{\coqin{X}} = \coqin{op'}_{\coqin{X}} \circ (\coqin{E} \mysharp{} \coqin{e}_{\coqin{X}}). 
$$

%In the case of algebraic operations, the lifting of an
%operation can be defined uniformly (this is what is meant by ``modular
%monad transformers'').

\tikzset{myblack/.style={->,draw=black}, myblue/.style={->,draw=blue}, myred/.style={->,draw=red},
  mygreen/.style={->,draw=green}, myviolet/.style={->,draw=violet}, myorange/.style={->,draw=orange}}
\def\myarrow#1{\begin{tikzpicture}[baseline=-0.5ex]\draw[#1,->] (0,0)--(0.3,0);\end{tikzpicture}}

\begin{theorem}[Uniform Algebraic Lifting {\cite[Thm.~19]{jaskelioff2009esop}}]
\label{thm:theorem19}
Given an algebraic \coqin{E}-operation \coqin{op} for \coqin{M} and a monad morphism \coqin{e} from \coqin{M} to \coqin{N}, let \coqin{op'} be
$$
\coqin{X} \mapsto \Join{N}{X} \circ \coqin{e}_{(\coqin{N X})} \circ \coqin{op}_{(\coqin{N X})} \circ (\coqin{E} \mysharp{} \Ret{M}{(N X)}).
$$
Then \coqin{op'} is an algebraic \coqin{E}-operation for \coqin{N} and a lifting of \coqin{op} along \coqin{e}.
\end{theorem}
\begin{proof}
The proof that \coqin{op'} is a lifting is depicted by the diagram of Fig.~\ref{fig:diag19}.

\def\lbl#1{${}^{#1}$}
\def\lbla{(a)}
\def\lblb{(b)}
\def\lblc{(c)}
\def\lbld{(d)}
\def\lble{(e)}
\def\lblf{(f)}
\def\lblg{(g)}
\def\lblh{(h)}
\def\lbli{(i)}
\def\lblj{(j)}
\def\lblk{(k)}
\def\lbll{(l)}

\begin{figure}[t]
\centering
\begin{tikzpicture}
\node (EMMX) {\coqin{E(M(M X))}} ;
\node (MMX) [right=of EMMX,xshift=2em] {\coqin{M(M X)}} ;
\node (EMNX) [above=of EMMX] {\coqin{E(M(N X))}} ;
\node (MNX) [above=of MMX] {\coqin{M(N X)}} ;

\node (EMX) [below left=of EMMX] {\coqin{E(M X)}};
\node (ENX) [above left=of EMNX] {\coqin{E(N X)}};
\node (MX) [below right=of MMX] {\coqin{M X}};
\node (NX) [above right=of MNX] {\coqin{N X}};

\path (EMX) edge[myblack] node[below] {\lbl{\lbll} $\coqin{op}_{\coqin{X}}$} (MX);
\path (MX) edge[myblack] node[right] {\lbl{\lblk} $\coqin{e}_{\coqin{X}}$} (NX);

\path (EMX) edge[myblue] node[left] {\lbl{\lble} $\coqin{E} \mysharp{} \coqin{e}_{\coqin{x}}$} (ENX) ;
\path (ENX) edge[myblue] node[right] {\lbl{\lblb} $\coqin{E} \mysharp{} \ret$} (EMNX) ;
\path (EMNX) edge[myblue] node[above] {\lbl{\lblc} $\coqin{op}_{\coqin{(N X)}}$} (MNX) ;
\path (MNX) edge[myblue] node[left] {\lbl{\lbld} $\join \circ \coqin{e}_{\coqin{(N X)}}$} (NX) ;

\path (EMX) edge[myred] node[right] {\lbl{\lblf} $\coqin{E} \mysharp{} \ret$} (EMMX) ;
\path (EMMX) edge[myred] node[left] {\lbl{\lblg} $\coqin{E} \mysharp{} (\coqin{M} \mysharp{} \coqin{e}_{\coqin{X}})$} (EMNX) ;

\path (EMMX) edge[mygreen] node[below] {\lbl{\lblh} $\coqin{op}_{\coqin{(M X)}}$} (MMX) ;
\path (MMX) edge[mygreen] node[right] {\lbl{\lbli} $\coqin{M} \mysharp{} \coqin{e}_{\coqin{X}}$} (MNX) ;

\path (MMX) edge[myviolet] node[left] {\lbl{\lblj} $\join$} (MX) ;

\path (ENX) edge[myorange] node[above] {\lbl{\lbla} lifting of $\coqin{op}_{\coqin{X}}$} (NX) ;
\end{tikzpicture}
\caption{Proof of Uniform Algebraic Lifting (Theorem~\ref{thm:theorem19})}
\label{fig:diag19}
\end{figure}
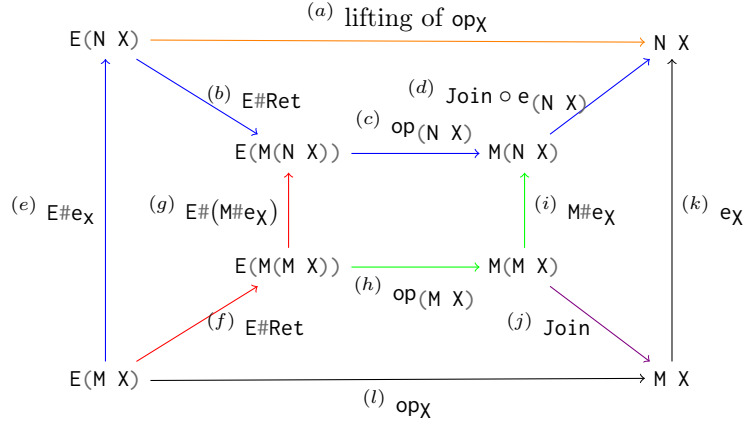

The first step is to show that the path \lbla{} (\myarrow{myorange}) and the 
path \lblb{}-\lblc{}-\lbld{} (\myarrow{myblue}) are equal, which is by definition of a lifting.
The resulting goal is rendered in \coq{} as follows (for any~\coqin{Y}):
\begin{minted}{ssr}
e X (op X Y) = Join (e (N X) (op (N X) ((E # Ret) ((E # e X) Y))))
\end{minted}

The second step of the proof is to show that the path \lble{}-\lblb{}
(\myarrow{myblue}) and the path \lblf{}-\lblg{}
(\myarrow{myred}) are equal, which is achieved by appealing to the functor laws
and the naturality of $\Ret{}{}$. More precisely, to prove
\begin{minted}{ssr}
(E # Ret) ((E # e X) Y) = (E # (M # e X)) ((E # Ret) Y),
\end{minted}
it suffices to execute the following sequence of rewritings:
\begin{minted}[fontsize=\small]{ssr}
rewrite -[in LHS]compE -functor_o. (* functor composition law in the lhs *)
rewrite -[in RHS]compE -functor_o. (* functor composition law in the rhs *)
(* the goal is now: (E # (Ret \o e X)) Y = (E # (M # e X \o Ret)) Y *)
rewrite (natural RET).             (* naturality of ret *)
(* the goal is now: (E # (Ret \o e X)) Y = (E # (Ret \o FId # e X)) Y *)
by rewrite FIdf.                   (* property of the identity functor *)
\end{minted}

The next step is to show that the path \lblg{}-\lblc{} and the
path \lblh{}-\lbli{} (\myarrow{mygreen}) are equal; this is by naturality of
\coqin{op}.

The next step is to show that the paths
\lbli{}-\lbld{} and
\lblj{}-\lblk{} are equal, which is by 
the bind law of monad morphisms and naturality of monad morphisms.

The last step (equality of the paths
\lblf{}-\lblh{}-\lblj{} and
\lbll{}) amounts to proving:
\begin{minted}{ssr}
op X Y = Join (op (M X) ((E # Ret) Y)).
\end{minted}
This step depends of an intermediate
lemma~\cite[Prop.~17]{jaskelioff2009esop}.  Let us explain it because
it introduces functions and we will use one of them again later in
this paper.
Given a natural transformation \coqin{n : E ~> M}, \coqin{psi} is an
\coqin{E}-operation for \coqin{M} defined by the function
$\coqin{X} \mapsto \Join{}{X} \circ \coqin{n}$.
Given an \coqin{E}-operation for \coqin{M}, \coqin{phi}
is a natural transformation \coqin{E ~> M} defined by the function
$\coqin{X} \mapsto \coqin{op}_{\coqin{X}} \circ (\coqin{E} \mysharp{} \Ret{}{})$.
It turns out that \coqin{psi} is algebraic and that \coqin{psi}
cancels \coqin{phi} for algebraic operations (proofs omitted here, see~\cite{monae}), which
proves the last goal.

The second part of the proof is to prove that \coqin{op'} is
algebraic.  This is a direct consequence of the fact that \coqin{psi}
is algebraic.

It should be noted that, even though the statement of the theorem
defines the lifting as the composition of the functions \coqin{Join},
\coqin{e}, etc., it is actually much more practical from the view
point of formal proof to define it as \coqin{psi (monadM_nt e \v phi
  op)}, i.e., the application of the function \coqin{psi} to the
vertical composition of \coqin{e} and \coqin{phi op}, because this
object (let us call it \coqin{alifting}) is endowed with the
properties of algebraic operations, whose immediate availability
facilitates the formal proof.
\end{proof}

The reader can observe in Appendix~\ref{appendix:theorem19} that the
complete proof script for Theorem \ref{thm:theorem19} essentially
amounts to a small number of rewritings, as has been partially
illustrated in the proof just above.

\subsection*{Example: Lifting the get Operation along the Exception Monad Transformer}

Let us assume the availability of a type \coqin{S} for states and of a
type \coqin{Z} for exceptions. We consider \coqin{M} to be the state
monad.  To define the lifting of the get operation of \coqin{M} (more
precisely its algebraic version seen in
Sect.~\ref{sec:algebraic_operations}) along \coqin{exceptT}
(Sect.~\ref{sec:transf}) it suffices to call the \coqin{alifting}
function with the right arguments:
\begin{minted}{ssr}
Let M S : monad := ModelState.state S.

Definition aLGet {Z S} : (StateOps.Get.func S).-aoperation (exceptT Z (M S)) :=
  alifting (get_aop S) (Lift (exceptT Z) (M S)).
\end{minted}
By the typing, we see that the result \coqin{aLGet} is also an
algebraic operation.

For example, we can check that the resulting sigma-operation is indeed
the get operation of the transformed monad:
\begin{minted}{ssr}
Goal forall Z (S : UU0) X (k : S -> exceptT Z (M S) X),
  aLGet _ k = StateOps.get_op _ k. by [].
\end{minted}

\section{Application 3: Formalization of the Lifting of Sigma-Operations}
\label{sec:lifting_operations}

This section is an application of our formalization of
sigma-operations and (functorial) monad transformers of
Sect.~\ref{sec:extension} and also of Theorem~\ref{thm:theorem27}.
Using \monae{}, we give an equational proof for a theorem that
generalizes the lifting of Sect.~\ref{sec:theorem19} which was
restricted to algebraic operations.
This corresponds to the second part of the original paper on modular
monad transformers~\cite[Sect.~5]{jaskelioff2009esop}.

This application requires us to use a non-standard setting of \coq{}.
Section~\ref{sec:codensity} introduces a monad transformer whose
formalization requires impredicativity.
Section~\ref{sec:naturality_of_m} focuses on the main technical
difficulty that we identified when going from the pencil-and-paper
proofs to a formalization using \coq{}: an innocuous-looking proof
that actually calls for an argument based on parametricity.
We conclude this section with the formal statement
of~\cite[Thm.~27]{jaskelioff2009esop} and its formal proof
(Sect.~\ref{sec:theorem27}).

\subsection{Impredicativity Setting for the Codensity Monad Transformer}
\label{sec:codensity}

To implement the lifting of an operation along a functorial monad
transformer, Jaskelioff introduces a monad transformer
\coqin{codensityT} related to the construction of the codensity monad
for an endofunctor~\cite[Def.~23]{jaskelioff2009esop}.  Its
formalization requires impredicativity and if nothing is done, the
standard setting of \coq{} would lead to \newterm{universe
  inconsistencies}.

Let us give a bit of background on impredicativity with \coq{}.
The type theory of \coq{} is constrained by a hierarchy of
universes \Set{}, $\coqin{Type}_1$, $\coqin{Type}_2$, etc.
The \coq{} language only provides the keywords \Set{} and
\coqin{Type}, the \coq{} system figures out the right indices for
\coqin{Type}s.
Universes are not impredicative by default; yet, \coq{} has
an option (\impredicativeset{}) that changes the logical theory by
declaring the universe \Set{} as impredicative.
This option is useful in \coq{} to formalize System $F$/$F\omega$,
their impredicative encodings of data types, and for extraction of
programs in CPS style.
It is known to be inconsistent with some standard axioms of classical
mathematics \cite{impredicativeset,geuvers2001} but we do not rely on
them here\footnote{More precisely, the development we discuss in this
  paper~\cite[directory \coqin{impredicative_set}]{monae} uses
  together with impredicative \Set{} only the standard axioms of
  functional extensionality and proof irrelevance, which are
  compatible.}.
To keep a firm grip on the universes involved, we fix a few universes
at the beginning of the formal development~\cite[file
\coqin{ihierarchy.v}]{monae}:
\begin{minted}[escapeinside=77]{ssr}
Definition UU2 : Type := Type.
Definition UU1 : UU2 := Type.
Definition UU0 : UU1 := 7\Set{}7.
\end{minted}
%(this is inspired by~\cite{voevodsky2015mscs})
and only use them instead of \Set{} or \coqin{Type} (so far we have
been using \coqin{UU0} but it is really another name for the native
\Set{} universe).

Now that we have set \coq{} appropriately, we define the codensity
monad transformer.
Given a monad~\coqin{M}, a computation of a value of type \coqin{A} in the monad
\coqin{codensityT M} has type \coqin{forall (B : UU0), (A -> M B) -> M B}
of type \coqin{UU0}: here, impredicativity comes into play.
We abbreviate this type expression as \coqin{MK M A} in the following.
We do not detail the formalization of \coqin{codensityT} because it
follows the model of the exception monad transformer that we explained in
Sect.~\ref{sec:transf}. Let us just display its main ingredients,
i.e., the unit, bind, and lift
operations~\cite[Def.~23]{jaskelioff2009esop}:
\begin{minted}{ssr}
Definition retK (A : UU0) (a : A) : MK M A :=
  fun (B : UU0) (k : A -> M B) => k a.
Definition bindK (A B : UU0) (m : MK M A) f : MK M B :=
  fun (C : UU0) (k : B -> M C) => m C (fun a : A => (f a) C k).
(* definition of codensityTmonadM omitted *)
Definition liftK (A : UU0) (m : M A) : codensityTmonadM A :=
  fun (B : UU0) (k : A -> M B) => m >>= k.
\end{minted}
We can check in \coq{} that they indeed give rise to a monad
transformer in the sense of Sect.~\ref{sec:transf}, so that
\coqin{codensityT} does have the type \coqin{monadT}
(Sect.~\ref{sec:transf}) of monad transformers.

\subsection{Parametricity to Prove Naturality}
\label{sec:naturality_of_m}

The monad transformer \coqin{codensityT} is needed to state the
theorem about the lifting of sigma-operations and in particular to
define a natural transformation called
\coqin{from}~\cite[Prop.~26]{jaskelioff2009esop}.
Formally, we can define \coqin{from}'s components as follows
(\coqin{M} is a monad):
\begin{minted}{ssr}
Definition from_component : codensityT M ~~> M :=
  fun (A : UU0) (c : codensityT M A) => c A Ret.
\end{minted}
At first sight, the naturality of \coqin{from_component} seems obvious
and indeed no proof is given in the original paper on modular monad
transformers~(see the first of the two statements
of~\cite[Prop.~26]{jaskelioff2009esop}). It is however a bit more
subtle than it appears and, as a matter of fact, it is shown in a
later paper that this claim is wrong: $\coqin{from}_{\coqin{M}}$
cannot be a natural transformation in the setting of $F\omega$
\cite[p.~4452]{jaskelioff2010}. We explain how we save the day in
\coq{} by relying on parametricity.

We state the naturality of $\coqin{from}_{\coqin{M}}$ as
\coqin{naturality (codensityT M) M from_component}.
This goal reduces\footnote{By functional extensionality, by naturality of
  {\tt Ret}, and by definition of {\tt from\us{}component}.} to:
\begin{minted}{ssr}
forall (m : codensityT M A) (h : A -> B), (M # h \o m A) Ret = m B (M # h \o Ret).
\end{minted}
This last goal is an instance of a more general statement
(recall from Sect.~\ref{sec:codensity} that \coqin{MK M} is the action on the objects
of the monad \coqin{codensityT M}):
\begin{minted}{ssr}
forall (M : monad) (A : UU0) (m : MK M A) (A1 B : UU0) (h : A1 -> B),
  M # h \o m A1 = m B \o (fun f : A -> M A1 => (M # h) \o f).
\end{minted}
This is actually a special case of naturality as one can observe by rewriting
the type of \coqin{m} with the appropriate functors: 
\coqin{exponential_F A \O M} and \coqin{M}, where \coqin{exponential_F} is
the functor whose action on objects is \coqin{forall X : UU0, A ->
  X}:
\begin{minted}{ssr}
forall (M : monad) (A : UU0) (m : MK M A), naturality (exponential_F A \O M) M m
\end{minted}
Unfortunately, we are not able to prove it in plain \coq{} (with
or without impredicative \Set{}), even if we consider particular functors
\coqin{M} such as the identity functor.

The solution consists in assuming an axiom of parametricity for each
functor~\coqin{M} and derive naturality from it.  That is, we follow
the approach advocated by Wadler~\cite{DBLP:conf/fpca/Wadler89}.  It
has been shown to be sound in
\coq{}~\cite{DBLP:conf/csl/KrishnaswamiD13,bernardy2012lics,keller2012csl,atkey2014popl}
and it is implemented by the \paramcoq{} plugin~\cite{keller2012csl}.
For instance, let us describe what happens when \coqin{M} is the list
monad. First, we rewrite the naturality statement above in the case of
the list functor (\coqin{map} is the map function of lists):
\begin{minted}{ssr}
forall (X Y : UU0) (f : X -> Y) (g : A -> seq X),
  (map f \o m X) g = (m Y \o (exponential_F A \O M) # f) g.
\end{minted}
The proof proceeds by induction on a proof-term of type
\begin{minted}{ssr}
list_R X Y (fun x y => f x = y) (m X g) ((m Y \o (exponential_F A \O M) # f) g)
\end{minted}
where \coqin{list_R X Y X_R l1 l2} means that the elements of
lists~\coqin{l1} and~\coqin{l2} are pairwise related by the relation
\coqin{X_R}.
The role of \paramcoq{} is to generate definitions (including
\coqin{list_R}) for us to be able to produce this proof.
Concretely, starting from \coqin{MK}, \paramcoq{} generates the
logical relation \coqin{T_R} of type (it is obtained by induction on
types~\cite{Goubault-LarrecqLN08}):
\begin{minted}{ssr}
(forall X : UU0, (A -> list X) -> list X) ->
  (forall X : UU0, (A -> list X) -> list X) -> UU0
\end{minted}
Here, \coqin{T_R m1 m2} expands to:
\begin{minted}{ssr}
forall (X1 X2 : UU0) (RX : X1 -> X2 -> UU0)
  (f1 : A -> list X1) (f2 : A -> list X2),
(forall a1 a2 : A, a1 = a2 -> list_R X1 X2 RX (f1 a1) (f2 a2)) ->
list_R X1 X2 RX (m1 X1 f1) (m2 X2 f2)
\end{minted}
It is then safe to assume the following parametricity axiom:
\begin{minted}{ssr}
Axiom param : forall m : MK M A, T_R m m.
\end{minted}
The application of \coqin{param} is the first step to produce the
proof required for the induction:
\begin{minted}[escapeinside=77]{ssr}
have : list_R X Y (fun x y = f x = y) (m X g) ((m Y \o (exponential_F A \O M) # f) g).
  apply: param.
  (* 7$\forall$7 a a', a = a' ->
    list_R X Y (fun x y = f x = y) (g a) (((exponential_F A \O M) # f) g a') *)
\end{minted}
The goal generated is proved by induction on \coqin{g a} which is a list.

The same approach is applied to other monads (identity, exception,
option, state)~\cite[file
\coqin{iparametricity_codensity.v}]{monae}.

\subsection{Lifting of Sigma-operations: Formal Statement}
\label{sec:theorem27}

Before stating and proving the main theorem about lifting of
sigma-operations, we formally define a special algebraic
operation~\cite[Def.~25]{jaskelioff2009esop}.
Let \coqin{E} be a functor, \coqin{M} be a monad, and \coqin{op} be an
\coqin{E}-operation for \coqin{M}. The natural transformation
\coqin{kappa} 
from \coqin{E} to \coqin{codensityT M}
is defined by the components
$$
\coqin{A}, (\coqin{s : E A}), \coqin{B}, (\coqin{k : A -> M B}) \mapsto
\coqin{op} \, \coqin{B} \, ((\coqin{E} \mysharp{} \coqin{k}) \, \coqin{s})
$$
and \coqin{psik} is the algebraic \coqin{E}-operation for the monad
\coqin{codensityT M} defined by:
\begin{minted}{ssr}
Definition psik : E.-aoperation (codensityT M) := psi (kappa op).
\end{minted}
Recall that the function \coqin{psi} has been defined in the proof of
Theorem~\ref{thm:theorem19}.

\begin{theorem}[Uniform Lifting {\cite[Thm.~27]{jaskelioff2009esop}}]
\label{thm:theorem27}
Let \coqin{M} be a monad such that  any computation
\coqin{m : MK M A} is natural in the sense of 
Sect.~\ref{sec:naturality_of_m} (hypothesis \coqin{naturality_MK}).
  Let \coqin{op} be an \coqin{E}-operation for \coqin{M} and
  \coqin{t} be a functorial monad transformer.
  We denote:
  \begin{itemize}
  \item by \coqin{op1} the term \coqin{Hmap t (from naturality_MK)}
    (see Sect.~\ref{sec:naturality_of_m} for \coqin{from}, \coqin{Hmap} was defined in Sect.~\ref{sec:fmt}),
  \item by \coqin{op2} the algebraic lifting along \coqin{Lift t} of
    \coqin{(psik op)} (see just above for \coqin{psik}), and
  \item by \coqin{op3} the term \mintinline{ssr}{E ## Hmap t
      (monadM_nt (Lift codensityT M))} (see Sect.~\ref{sec:codensity}
    for \coqin{codensityT}).
  \end{itemize}
  Then the operation \coqin{op1 \v op2 \v op3} (where \coqin{\v} is the vertical composition seen in Sect.~\ref{sec:background}) is a lifting of \coqin{op} along \coqin{t}.
\end{theorem}
\begin{proof}
The proof is depicted by the diagram in Fig.~\ref{fig:diag27}.

\def\gbl#1{${}^{#1}$}
\def\gbla{(a)}
\def\gblb{(b)}
\def\gblc{(c)}
\def\gbld{(d)}
\def\gble{(e)}
\def\gblf{(f)}
\def\gblg{(g)}
\def\gblh{(h)}
\def\gbli{(i)}
\def\gblj{(j)}
\def\gblk{(k)}
\def\gbll{(l)}

\begin{figure}[t]
\centering
\begin{tikzpicture}
\node (EKMX)                  {\coqin{E( K M X )}} ;
\node (KMX) [right=of EKMX]   {\hspace{1em}\coqin{K M X}} ;
\node (ETKMX) [above=of EKMX] {\coqin{E(T K M X)}} ;
\node (TKMX) [above=of KMX]   {\coqin{T K M X}} ;

\node (EMX) [below left=of EKMX,xshift=-1.8cm]   {\coqin{E( M X )}};
\node (ETMX) [above left=of ETKMX,xshift=-1.8cm] {\coqin{E(T M X)}};
\node (MX) [below right=of KMX,xshift=1.8cm]    {\hspace{1em}\coqin{M X}};
\node (TMX) [above right=of TKMX,xshift=1.8cm]  {\coqin{T M X}};

\path (EMX) edge[myblack] node[below] {\footnotesize \gbl{\gbll} $\coqin{op}_{\coqin{X}}$} (MX);
\path (MX) edge[myblack] node[above,rotate=-90] {\footnotesize \gbl{\gblk} $(\coqin{Lift t M})_{\coqin{X}}$} (TMX);

\path (EMX) edge[myblue] node[above,rotate=90] {\footnotesize \gbl{\gble} $\coqin{E} \mysharp{} (\coqin{Lift t M})_{\coqin{X}}$} (ETMX) ;
\path (ETMX) edge[myblue] node[right] {\footnotesize \gbl{\gblb} $\coqin{op3}_{\coqin{X}}$} (ETKMX) ;
\path (ETKMX) edge[myblue] node[above] {\footnotesize \gbl{\gblc} $\coqin{op2}_{\coqin{X}}$} (TKMX) ;
\path (TKMX) edge[myblue] node[left] {\footnotesize \gbl{\gbld} $\coqin{op1}_{\coqin{X}}$} (TMX) ;

\path (EMX) edge[myred] node {\footnotesize \gbl{\gblf} $\coqin{E} \mysharp{} (\coqin{Lift codensityT M})_{\coqin{X}}$} (EKMX) ;
\path (EKMX) edge[myred] node[left] {\footnotesize \gbl{\gblg} \scriptsize $\coqin{E} \mysharp{} (\coqin{Lift t (codensityT M)})_{\coqin{X}}$} (ETKMX) ;

\path (EKMX) edge[mygreen] node[below] {\footnotesize \gbl{\gblh} $(\coqin{psik op})_{\coqin{X}}$} (KMX) ;
\path (KMX) edge[mygreen] node[right] {\footnotesize \gbl{\gbli} \scriptsize $(\coqin{Lift t (codensityT M)})_{\coqin{X}}$} (TKMX) ;

\path (KMX) edge[myviolet] node[left] {\footnotesize \gbl{\gblj} $\coqin{from_nt}_{\coqin{X}}$} (MX) ;

\path (ETMX) edge[myorange] node[above] {\footnotesize \gbl{\gbla} lifting of $\coqin{op}_{\coqin{X}}$} (TMX) ;
\end{tikzpicture}
\caption{Proof of Uniform Lifting (Theorem~\ref{thm:theorem27})}
\label{fig:diag27}
\end{figure}
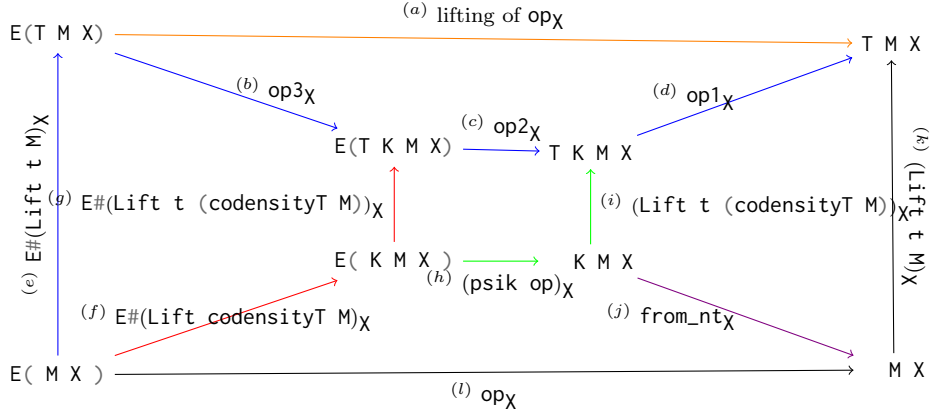

The first step of the proof is to unfold the definition of lifting
(which amounts to showing that the paths \gbla{} (\myarrow{myorange}) and
\gblb{}-\gblc{}-\gbld{} are equal). Consequently, the proof goal is rendered
in \coq{} as follows (for all \coqin{X : UU0}):
\begin{minted}{ssr}
Lift t M X \o op X = (op1 \v op2 \v op3) X \o E # Lift t M X
\end{minted}
The second step of the proof is to show that the path \gble{}-\gblb{}
and the path \gblf{}-\gblg{} (\myarrow{myred}) are equal, which is achieved by appealing
to the law of functor composition and the naturality of \coqin{Hmap}.

The next step is to show that the path \gblg{}-\gblc{}
and the path \gblh{}-\gbli{} (\myarrow{mygreen}) are equal; this is by 
applying Theorem~\ref{thm:theorem19}.

At this point, the goal becomes:
\begin{minted}{ssr}
Lift t M X \o op X = (op1 X \o (Lift t (codensityT M) X \o psik op X))
                     \o E # Lift codensityT M X
\end{minted}

It happens that we can use the naturality of \coqin{Hmap} to make
the \coqin{from} function appear in the right-hand side of the goal:
\begin{minted}{ssr}
Lift t M X \o op X = ((Lift t M X \o from naturality_MK X)
                     \o psik op X) \o E # Lift codensityT M X
\end{minted}

The last step is to identify \coqin{op} with the composition of
the \coqin{from} function, \coqin{psik op}, and \mintinline{ssr}{E # List codensityT M},
which is the purpose of a lemma~\cite[Prop.~26]{jaskelioff2009esop}
(see \cite[file \coqin{ifmt_lifting.v}, lemma \coqin{psikE}]{monae}).
\end{proof}

The proof script corresponding to the proof above is reproduced in
Appendix~\ref{appendix:theorem27}.

Finally, we show that, for all the monad transformers considered in
this paper, the lifting of an algebraic operation provided by
Theorem~\ref{thm:theorem27} coincides with the one provided by
Theorem~\ref{thm:theorem19}. This corresponds to the last results
about modular monad transformers~\cite[Prop.~28]{jaskelioff2009esop}.

\section{Related Work}
\label{sec:related_work}

The example we detail in Sect.~\ref{sec:application} adds to several
examples of monadic equational
reasoning~\cite{gibbons2011icfp,gibbons2012utp,mu2019tr2,mu2019tr3,pauwels2019mpc,mu2020flops}.
Its originality is to use a parameterized interface and the
\coqin{RunStateT} command, which are typical of programs written using
monad transformers.

Huffman formalizes three monad transformers in the Isabelle/HOL proof
assistant~\cite{huffman2012icfp}.
This experiment is part of a larger effort to overcome the limitations
of Isabelle/HOL type classes to reason about Haskell programs that use
(Haskell) type classes.
Compared to Isabelle/HOL, the type system of \coq{} is more expressive
so that we could formalize a much larger theory, even relying on extra
features of \coq{} such as impredicativity and parametricity to do so.

Maillard proposes a meta language to define monad transformers in the
\coq{} proof assistant~\cite[Chapter~4]{maillard2019phd}.
It is an instance implementation of one element of a larger framework
to verify programs with monadic effects using Dijskstra
monads~\cite{maillard2019icfp}.
The lifting of operations is one topic of this framework but it does
not go as far as the deep analysis of
Jaskelioff~\cite{jaskelioff2009phd,jaskelioff2009esop,jaskelioff2010}.

There are also formalizations of monads and their morphisms that focus
on the mathematical aspects, e.g., UniMath~\cite{UniMath}. However,
the link to the monad transformers of functional programming is not
done.

Monad transformers is one approach to combine effects.  Algebraic
effects is a recent alternative. It turns out that the two are
related~\cite{schrijvers2019haskell} and we have started to extend
\monae{} to clarify formally this relation.

\section{Conclusions and Future Work}
\label{sec:conclusion}

In this paper, we extended \monae{}, a formalization of monadic
equational reasoning, with monad transformers. We explained how it
helps us to better organize the models of monads, thanks to
sigma-operations in particular. We also explained how to extend the
hierarchy of monad interfaces to handle programs written with monad
transformers in mind. We also used our formalization of monad
transformers to formalize the theory of liftings of modular monad
transformers~\cite{jaskelioff2009esop} using equational reasoning. For
that purpose, we needed to fix the original presentation by using
\coq{}'s impredicativity and parametricity.

The main result of this paper is a robust, formal theory of monad
transformers. We plan to extend the hierarchy of monad interfaces of
\monae{} similarly to how we proceeded for
\coqin{exceptStateRunMonad}. Such an extension will call for more
models to be formalized and we expect our formalized theory of
liftings to be useful on this occasion.

Results up to Sect.~\ref{sec:theorem19} hold whether or not \Set{} is
impredicative. In contrast, the setting of
Sect.~\ref{sec:lifting_operations} conflicts with \monae{} programs
relying on some data structures from the \mathcomp{}
library~\cite{mathcomp} (such as fixed-size lists) or from the
\infotheo{} library~\cite{infotheo} (such as probability
distributions) because these data structures are in \coqin{Type} and
cannot be computed with monads in~\Set{}.
One could think about reimplementing them but this is a substantial
amount of work.
A cheap way to preserve these data structures together with the
theorem on lifting of sigma-operations is to disable universe checking
as soon as this theorem is used; this way, monads can stay
in~\coqin{Type}.
Disabling universe checking is not ideal because it is unsound in
general\footnote{One can derive \coqin{False} by applying a variant of
  Hurkens paradox (see
  \url{https://coq.inria.fr/library/Coq.Logic.Hurkens.html}).}; note
however that this is sometimes used for the formalization of
category-theoretic notions (e.g., \cite[Sect.~6]{ahrens2017csl}).
How to improve this situation is another direction for future work.

\paragraph{Acknowledgements}
We acknowledge the support of the JSPS KAKENHI Grant Number 18H03204.
We thank all the participants of the JSPS-CNRS bilateral program
``FoRmal tools for IoT sEcurity'' (PRC2199) for fruitful discussions.
We also thank Takafumi Saikawa for his comments. This work is based on
joint work with C\'elestine Sauvage~\cite{sauvage2020jfla}.

\appendix

\section{Proof Scripts for the Three Applications of This Paper}

The following proof scripts have been copied verbatim from
\monae{}~\cite{monae} for the reader's convenience.
We claim that these proof scripts are readable and short in the sense
that each line corresponds to a genuine proof step and that there are
few administrative tactics hampering reading.
The terseness of the \ssreflect{} tactic language could actually make
these proof scripts much shorter but that is not our point here.
In particular, we make explicit the proof steps of
Theorems~\ref{thm:theorem19} and~\ref{thm:theorem27} (in
Appendices~\ref{appendix:theorem19} and \ref{appendix:theorem27}) with
\coqin{transitivity} steps or explicit \coqin{rewrite} steps followed
by indented (sub-)proof scripts.

\subsection{Proof Script for the Correctness of \coqin{fastProduct}}
\label{sec:fpscript}

The following proof script can be found in~\cite[file
\coqin{example_transformer.v}]{monae}.

\begin{minted}{ssr}
Lemma fastProductCorrect l n :
  evalStateT (fastProduct l) n = Ret (product l).
Proof.
rewrite /fastProduct -(mul1n (product _)); move: 1.
elim: l => [ | [ | x] l ih] m.
- rewrite muln1 bindA bindretf putget.
  rewrite /evalStateT RunStateTCatch RunStateTBind RunStateTPut bindretf.
  by rewrite RunStateTRet RunStateTRet catchret bindretf.
- rewrite muln0.
  rewrite /evalStateT RunStateTCatch RunStateTBind RunStateTBind RunStateTPut.
  by rewrite bindretf RunStateTFail bindfailf catchfailm RunStateTRet bindretf.
- rewrite [fastProductRec _]/=.
  by rewrite -bindA putget bindA bindA bindretf -bindA -bindA putput ih mulnA.
Qed.
\end{minted}

\subsection{Proof Script for Theorem~\ref{thm:theorem19} \cite[Thm.~19]{jaskelioff2009esop}}
\label{appendix:theorem19}

The following proof script can be found in~\cite[file
\coqin{monad_transformer.v}]{monae}.

\begin{minted}{ssr}
Section uniform_algebraic_lifting.
Variables (E : functor) (M : monad) (op : E.-aoperation M).
Variables (N : monad) (e : monadM M N).

Definition alifting : E.-aoperation N := psi (monadM_nt e \v phi op).

Lemma aliftingE :
  alifting = (fun X => Join \o e (N X) \o phi op (N X)) :> (_ ~~> _).
Proof. by []. Qed.

Theorem uniform_algebraic_lifting : lifting op e alifting.
Proof.
move=> X.
apply fun_ext => Y.
rewrite /alifting !compE psiE vcompE phiE !compE.
rewrite (_ : (E # Ret) ((E # e X) Y) =
             (E # (M # e X)) ((E # Ret) Y)); last first.
  rewrite -[in LHS]compE -functor_o.
  rewrite -[in RHS]compE -functor_o.
  rewrite (natural RET).
  by rewrite FIdf.
rewrite (_ : op (N X) ((E # (M # e X)) ((E # Ret) Y)) =
             (M # e X) (op (M X) ((E # Ret) Y))); last first.
  rewrite -(compE (M # e X)).
  by rewrite (natural op).
transitivity (e X (Join (op (M X) ((E # Ret) Y)))); last first.
  rewrite joinE monadMbind.
  rewrite bindE -(compE _ (M # e X)).
  by rewrite -natural.
by rewrite -[in LHS](phiK op).
Qed.
End uniform_algebraic_lifting.
\end{minted}

\subsection{Proof Script for Theorem~\ref{thm:theorem27} \cite[Thm.~27]{jaskelioff2009esop}}
\label{appendix:theorem27}

The following proof script can be found in~\cite[file
\coqin{ifmt_lifting.v}]{monae}.

\begin{minted}{ssr}
Section uniform_sigma_lifting.
Variables (E : functor) (M : monad) (op : E.-operation M) (t : FMT).
Hypothesis naturality_MK : forall (A : UU0) (m : MK M A),
  naturality_MK m.

Let op1 : t (codensityT M) ~> t M := Hmap t (from naturality_MK).
Let op2 := alifting (psik op) (Lift t _).
Let op3 : E \O t M ~> E \O t (codensityT M) :=
  E ## Hmap t (monadM_nt (Lift codensityT M)).

Definition slifting : E.-operation (t M) := op1 \v op2 \v op3.

Theorem uniform_sigma_lifting : lifting_monadT op slifting.
Proof.
rewrite /lifting_monadT /slifting => X.
apply/esym.
transitivity ((op1 \v op2) X \o op3 X \o E # Lift t M X).
  by rewrite (vassoc op1).
rewrite -compA.
transitivity ((op1 \v op2) X \o
  ((E # Lift t (codensityT M) X) \o (E # Lift codensityT M X))).
  congr (_ \o _); rewrite /op3.
  by rewrite -functor_o -natural_hmap functor_o functor_app_naturalE.
transitivity (op1 X \o
  (op2 X \o E # Lift t (codensityT M) X) \o E # Lift codensityT M X).
  by rewrite vcompE -compA.
rewrite -uniform_algebraic_lifting.
transitivity (Lift t M X \o from naturality_MK X \o (psik op) X \o
  E # Lift codensityT M X).
  congr (_ \o _).
  by rewrite compA natural_hmap.
rewrite -2!compA.
congr (_ \o _).
by rewrite compA -psikE.
Qed.
End uniform_sigma_lifting.
\end{minted}

\end{document}